\documentclass[a4paper,UKenglish]{lipics-v2018}

\usepackage{microtype}
\usepackage{todonotes}

\nolinenumbers 

\title{An Interesting Structural Property Related to the Problem of Computing All the Best Swap Edges of a Tree Spanner in Unweighted Graphs}

\titlerunning{Property Related to the ABSE Problem of a Tree Spanner in Unweighted Graphs}

\author{Davide Bil\`o}{Department of Humanities and Social Sciences, University of Sassari, Italy}{davidebilo@uniss.it}{https://orcid.org/0000-0003-3169-4300}{}

\author{Kleitos Papadopoulos}{InSPIRE, Agamemnonos 20, Nicosia, 1041, Cyprus}{kleitospa@gmail.com}{https://orcid.org/0000-0002-7086-0335}{}

\authorrunning{D. Bil\`o and K. Papadopoulos}

\Copyright{Davide Bil\`o and Kleitos Papadopoulos}



\category{}

\relatedversion{}

\supplement{}

\funding{}

\acknowledgements{}

\EventEditors{John Q. Open and Joan R. Access}
\EventNoEds{2}
\EventLongTitle{42nd Conference on Very Important Topics (CVIT 2016)}
\EventShortTitle{CVIT 2016}
\EventAcronym{CVIT}
\EventYear{2016}
\EventDate{December 24--27, 2016}
\EventLocation{Little Whinging, United Kingdom}
\EventLogo{}
\SeriesVolume{42}
\ArticleNo{23}

\begin{document}

\maketitle

\begin{abstract}
In this draft we prove an interesting structural property related to the problem of computing {\em all the best swap edges} of a {\em tree spanner} in unweighted graphs. Previous papers show that the maximum stretch factor of the tree where a failing edge is temporarily swapped with any other available edge that reconnects the tree depends only on the {\em critical edge}. However, in principle, each of the $O(n^2)$ swap edges, where $n$ is the number of vertices of the tree, may have its own critical edge. In this draft we show that there are at most 6 critical edges, i.e., each tree edge $e$ has a {\em critical set} of size at most 6 such that, a critical edge of each swap edge of $e$ is contained in the critical set.
\end{abstract}

\section{Basic definitions}
Let $G = (V, E)$ be a $2$-edge-connected undirected graph of $n$ vertices. Given an edge $e \in E$, we denote by $G-e=(V,E\setminus\{e\})$ the graph obtained after the removal of $e$ from $G$. Let $T$ be a tree spanning $V$ which is also a subgraph of $G$. Given an edge $e$ of $T$, let $S_e$ be the set of all the \emph{swap edges} for $e$, i.e., all edges in $E \setminus \{ e \}$ whose endpoints lie in two different connected components of $T-e$.
For any edge $e$ of $T$ and $f \in S_e$, let $T_{e/f}$ denote the \textit{swap tree} obtained from $T$ by replacing $e$ with $f$. Given two vertices $x,y \in V$, we denote by $d_G(x,y)$ the \emph{distance} between $x$ and $y$ in $G$, i.e., the number of edges contained in a shortest path in $G$ between $x$ and $y$. We define the \textit{stretch factor} $\sigma_G(T)$ of $T$ w.r.t. $G$ as 
\begin{equation*}
\sigma_G(T) = \max_{x,y \in V}  \frac{d_T(x,y)}{d_G(x,y)}.
\end{equation*}

\begin{definition}[Best Swap Edge]
Let $e$ be an edge of $T$. An edge $f^* \in S_e$ is a \textit{best swap edge} for $e$ if $f^* \in \arg\min_{f \in S_e} \sigma_{G-e}\big(T_{e/f}\big)$.
\end{definition}
For a swap edge $f=(x,y) \in S_e$ we say that $g=(a,b) \in S_e$ is {\em critical} for $f$ if $\sigma_S(T_{e/f})=d_T(x,a)+1+d_T(b,y)$. A set $C$ is {\em critical} for $e$ if, for every swap edge $f \in S_e$, $C$ contains a critical edge for $f$.

\section{Our result}
In this draft we show that there are at most 6 critical edges, i.e., each tree edge $e$ has a {\em critical set} of size at most 6 such that, a critical edge of each swap edge of $e$ is contained in the critical set.

Let $e$ be any fixed edge of $T$. Let $X$ be the set of vertices contained in one of the two connected components of $T-e$ (ties are chosen arbitrarily). Let $Y=V(G) \setminus X$ be the vertices contained in the other connected component of $T-e$. For every $x \in X$ and for every two (not necessarily distinct) edges $g=(a,b),g'=(a',b') \in S_e$, with $a,a' \in X$ and $b,b' \in Y$, we define 
\[
\phi_x(g,g'):=d_T(x,a)+d_T(b,b')+d_T(a',x).
\]

We denote by $(g_x,g'_x)$ a pair of (swap) edges in $\arg\max_{(g,g') \in S(e)^2}\phi_x(g,g')$. We use the notation $g_x=(a_x,b_x)$ and $g'_x=(a'_x,b'_x)$, with $a_x,a'_x \in X$ and $b_x,b'_x \in Y$. We denote by $T_X$ the subtree of $T$ induced by the vertices in $X$. Finally, for every edge $e'=(x,y)$ of $T_X$, we denote by $U(e',x)$ and $U(e',y)$ the partition of the vertices $X$ induced by $T_X - e'$ and containing $x$ and $y$, respectively. In~\cite{DBLP:conf/sirocco/BiloCG0P15} it is shown that for every $f=(x,y) \in S_e$, a critical edge for $f$ is either $g_x$ or $g'_x$.

\begin{lemma}\label{lm:useful_lemma}
Let $e'=(x,z) \in E(T_X)$ such that $\phi_z(g_x,g'_x) < \phi_z(g_z,g'_z)$. Then, one of the following two conditions is satisfied:
\begin{itemize}
\item  $a_{z}, a'_{z} \in U(e',x)$ and $\{a_x,a'_x\} \cap U(e',z) \neq \emptyset$;
\item $a_x,a'_x \in U(e',z)$ and $\{a_z,a'_z\} \cap U(e',x) \neq \emptyset$.
\end{itemize}
\end{lemma}
\begin{proof}
We first show that either $a_{z}, a'_{z} \in U(e',x)$ or $a_x,a'_x \in U(e',z)$ (or even both conditions) must hold. For the sake of contradiction, assume this is not the case. W.l.o.g., let $a_x,a_z \in U(e',x)$ and $a'_x,a'_z \in U(e',z)$. In this case, $\phi_z(g_z,g'_z) = \phi_x(g_z,g'_z) \leq \phi_x(g_x,g'_x) = \phi_z(g_z,g'_z)$. Therefore, $\phi_z(g_z,g'_z)$ cannot be stricly smaller that $\phi_z(g_x,g'_x)$.

Now we show that $\{a_x,a'_x\} \not \subseteq U(e',z)$. For the sake of contradiction, assume that $a_x,a'_x \in U(e',x)$. We have that
$\phi_z(g_z,g'_z) \leq \phi_x(g_z,g'_z)+2 \leq \phi_{x}(g_x,g'_x) + 2 = \phi_z(g_x,g'_x)$.
Therefore, $\phi_z(g_z,g'_z)$ cannot be stricly smaller that $\phi_z(g_x,g'_x)$ in this case. 

Finally, we show that $\{a_z,a'_z\} \not \subseteq U(e',x)$. For the sake of contradiction, assume that $a_z,a'_z \in U(e',z)$. From
$\phi_z(g_x,g'_x) < \phi_z(g_z,g'_z)$ we derive $\phi_x(g_x,g'_x) \leq \phi_z(g_x,g'_x)+2 < 
\phi_z(g_z,g'_z)+2 = \phi_x(g_x,g'_x)$, thus contradicting the choice of $g_x$ and $g'_x$. 
The claim follows.
\end{proof}

\begin{theorem}[6-critical-set theorem]\label{thm:6_critical_set}
For every edge $e$ of $T$, there exists a critical set of $e$ having size at most 6.
\end{theorem}
\begin{proof}
Let $e$ be any fixed edge of $T$. Let $e'=(x,z) \in E(T_X)$ such that $\phi_z(g_x,g'_x) < \phi_z(g_z,g'_z)$. If such an edge does not exist, then the critical set of $e$ has size at most 2. Therefore, we assume that such an edge exists. From Lemma~\ref{lm:useful_lemma}, one of the following two conditions is satisfied:
\begin{itemize}
\item  $a_{z}, a'_{z} \in U(e',x)$ and $\{a_x,a'_x\} \cap U(e',z) \neq \emptyset$;
\item $a_x,a'_x \in U(e',z)$ and $\{a_z,a'_z\} \cap U(e',x) \neq \emptyset$.
\end{itemize}
W.l.o.g., we assume that $a_{z}, a'_{z} \in U(e',x)$ and $\{a_x,a'_x\} \cap U(e',z) \neq \emptyset$. We prove that for every $z' \in U(e',z)$, $\{g_{z'},g'_{z'}\}=\{g_z,g'_z\}$. Indeed, $\phi_{z'}(g_{z'},g'_{z'}) \leq \phi_z(g_{z'},g'_{z'})+2d_T(z,z') \leq \phi_z(g_z,g'_z)+2d_T(z,z')=\phi_{z'}(g_z,g'_z)$. As a consequence, if for every vertex $x' \in U(e',x)$, $\{g_{x'},g'_{x'}\}=\{g_x,g'_x\}$, then $\{g_x,g'_x,g_z,g'_z\}$ would be a critical set of $e$ of size at most 4. Therefore, we only need to prove the claim when there exists a vertex $y \in U(e',x)$ such that $\{g_y,g'_y\} \neq \{g_x,g'_x\}$.

Let $x' \in U(e',x)$ be the vertex closest to $x$ such that $\{g_{x'},g'_{x'}\} = \{g_x,g'_x\}$ and there exists a neighbor $y$ of $x'$ in $U(e',x)$ such that $\phi_y(g_x,g'_x) < \phi_y(g_y,g'_y)$. Let $e''=(x',y)$. From Lemma~\ref{lm:useful_lemma}, one of the following two conditions is satisfied:
\begin{itemize}
\item  $a_{y}, a'_{y} \in U(e'',x')$ and $\{a_x,a'_x\} \cap U(e'',y) \neq \emptyset$;
\item $a_x,a'_x \in U(e'',y)$ and $\{a_y,a'_y\} \cap U(e'',x') \neq \emptyset$.
\end{itemize}
Since $U(e',z) \subseteq U(e'',x')$, we have that $\{a_x,a'_x\} \cap U(e'',x')\neq \emptyset$. Therefore, we can exclude the second of the two conditions and claim that $a_{y}, a'_{y} \in U(e'',x')$ as well as $\{a_x,a'_x\} \cap U(e'',y) \neq \emptyset$. We prove that for every $y' \in U(e'',y)$, $\{g_{y'},g'_{y'}\}=\{g_y,g'_y\}$. Indeed, $\phi_{y'}(g_{y'},g'_{y'}) \leq \phi_y(g_{y'},g'_{y'})+2d_T(y,y') \leq \phi_y(g_y,g'_y)+2d_T(y,y')=\phi_{y'}(g_y,g'_y)$. As a consequence, if for every vertex $x' \in U(e',x) \cap U(e'',x')$, $\{g_{x'},g'_{x'}\}=\{g_x,g'_x\}$, then $\{g_x,g'_x,g_z,g'_z,g_y,g'_y\}$ would be a critical set of $e$ of size at most 6. Therefore, we only need to prove the claim when there exists a vertex $t \in U(e',x)\cap U(e'',x')$ such that $\{g_t,g'_t\} \neq \{g_x,g'_x\}$. We conclude the proof by showing that such a vertex cannot exist. For the sake of contradiction, let $x'' \in U(e',x)\cap U(e'',x')$ be the vertex that minimizes the sum of distances from itself to both $x$ and $x'$ such that $\{g_{x''},g'_{x''}\} = \{g_x,g'_x\}$ and there exists a neighbor $t$ of $x''$ in $U(e',x) \cap U(e'',x')$ such that $\phi_t(g_x,g'_x) < \phi_{t}(g_{t},g'_{t})$. Let $e'''=(x'',t)$. From Lemma~\ref{lm:useful_lemma}, one of the following two conditions is satisfied:
\begin{itemize}
\item  $a_{t}, a'_{t} \in U(e''',x'')$ and $\{a_x,a'_x\} \cap U(e''',t) \neq \emptyset$;
\item $a_x,a'_x \in U(e''',t)$ and $\{a_y,a'_y\} \cap U(e''',x'') \neq \emptyset$.
\end{itemize}
Since $U(e',z),U(e'',y) \subseteq U(e''',x'')$, then $a_x,a'_x \in U(e''',x'')$. As a consequence, none of the two conditions can be satisfied. Hence, $t$ does not exist. This completes the proof.
\end{proof}




\bibliography{bibliography}

\begin{thebibliography}{1}

\bibitem{DBLP:conf/sirocco/BiloCG0P15}
Davide Bil{\`{o}}, Feliciano Colella, Luciano Gual{\`{a}}, Stefano Leucci, and
  Guido Proietti.
\newblock A faster computation of all the best swap edges of a tree spanner.
\newblock In Christian Scheideler, editor, {\em Structural Information and
  Communication Complexity - 22nd International Colloquium, {SIROCCO} 2015,
  Montserrat, Spain, July 14-16, 2015, Post-Proceedings}, volume 9439 of {\em
  Lecture Notes in Computer Science}, pages 239--253. Springer, 2015.
\newblock URL: \url{https://doi.org/10.1007/978-3-319-25258-2_17}.

\end{thebibliography}

\end{document}